\documentclass[11pt]{article}

\usepackage[T1]{fontenc}
\usepackage[utf8]{inputenc}
\usepackage{arxiv}
\usepackage{amsmath,amssymb,amsthm,mathtools}
\usepackage{array,booktabs,tabularx}
\usepackage{enumitem}
\usepackage{microtype}
\usepackage{xurl}
\usepackage[hidelinks]{hyperref}

\setlist{itemsep=0.25em,topsep=0.5em}
\hypersetup{
  pdftitle={Costs of Arbitrary Real Matrix Factorizations for Pure-DP Continual Counting},
  pdfauthor={Awnon Bhowmik; Mahmudul Hasan},
  pdfsubject={Prefix matrices, p-nuclearity, and pure differential privacy},
  pdfkeywords={continual counting, differential privacy, matrix mechanisms,
    operator ideals, p-nuclearity, approximation numbers}
}

\newcommand{\R}{\mathbb R}
\newcommand{\ctwo}{c_{2}}
\newcommand{\cfrob}{c_{\mathrm F}}
\newcommand{\gtwo}{\gamma_{2}}
\newcommand{\gfrob}{\gamma_{\mathrm F}}
\newcommand{\nucpow}{\mathfrak n}
\newcommand{\nucideal}{\mathfrak N}
\newcommand{\eps}{\varepsilon}
\newcommand{\dist}{\operatorname{dist}}
\newcommand{\rank}{\operatorname{rank}}
\newcommand{\MaxSE}{\operatorname{MaxSE}}
\newcommand{\MeanSE}{\operatorname{MeanSE}}
\newcolumntype{Y}{>{\raggedright\arraybackslash}X}
\newcolumntype{P}[1]{>{\raggedright\arraybackslash}p{#1}}

\newtheorem{theorem}{Theorem}[section]
\newtheorem{lemma}[theorem]{Lemma}
\newtheorem{proposition}[theorem]{Proposition}
\newtheorem{corollary}[theorem]{Corollary}
\theoremstyle{definition}
\newtheorem{definition}[theorem]{Definition}
\theoremstyle{remark}
\newtheorem{remark}[theorem]{Remark}

\numberwithin{equation}{section}

\title{Costs of Arbitrary Real Matrix Factorizations\\
for Pure-DP Continual Counting}
\author{%
  Awnon Bhowmik\\
  \small Department of Engineering and Computer Science\\
  \small Colorado Technical University\\
  \small \texttt{awnonbhowmik@outlook.com}
  \And
  Mahmudul Hasan\\
  \small Department of Mathematics\\
  \small University of Dhaka\\
  \small \texttt{mahmudul-2016713604@math.du.ac.bd}
}
\date{}
\renewcommand{\shorttitle}{Arbitrary Real Factorization Costs for Continual Counting}

\begin{document}
\maketitle

\begin{abstract}
Let \(T_n\) be the lower-triangular prefix-sum matrix and let
\(\cfrob(T_n)\) and \(\ctwo(T_n)\) be the factorization costs that govern mean
and maximum per-coordinate squared error of the Laplace matrix mechanism under
pure \(\eps\)-differential privacy, for \(\eps>0\).  We prove
\(\cfrob(T_n),\ctwo(T_n)=\Theta\bigl((\log(n+1))^{3/2}\bigr)\)
with no sign, sparsity, or squareness restriction and with arbitrary finite
inner dimension.  Consequently, within the pure-\(\eps\)-DP
matrix-mechanism class the optimized maximum and mean squared errors are both
\(\Theta(\eps^{-2}\log^{3}(n+1))\).  Under the factorization contract of
Arkhipov and Kalinin (arXiv:2607.08963v1), who prove the matching lower order
for factors with entries in \(\{0,1\}\) and state the arbitrary-factor
extension as open, the theorem below establishes the order for arbitrary real
factors.  The lower bound runs through a \(p\)-nuclear obstruction: an
aggregate column-width
estimate \(D_k(T_n)\asymp n^{3/2}k^{-1/2}\), valid in the low-rank range
\(1\leq k\leq n/16\), for the prefix chain, fed into the
classical approximation-space conversion of Pietsch and
Hinrichs--Pietsch, becomes harmonic at the critical exponent \(p=2/3\), and
Hölder's inequality transfers it to both factorization costs.  The same
computation determines \(\nucpow_p(T_n)\) for each fixed \(0<p<1\): order \(n\)
below \(2/3\), \(n\log n\) at \(2/3\), and \(n^{3p/2}\) above.  A Fenwick
interval factorization supplies matching upper bounds.  The claims are confined
to pure-\(\eps\)-DP Laplace matrix mechanisms and the two stated squared-error
criteria; they do not cover non-matrix continual mechanisms, approximate-DP
sensitivity, or expected maxima across coordinates.
\end{abstract}


\section{Introduction}

Continual counting releases every prefix sum of a stream while protecting each
update~\cite{dwork-continual,chan-shi-song}.  For a horizon \(n\), the workload
is the lower-triangular matrix
\begin{equation}
  T_n(i,j)=\mathbf 1\{j\leq i\},
  \qquad 1\leq i,j\leq n.
  \label{eq:prefix-matrix}
\end{equation}
A matrix mechanism~\cite{li-matrix-mechanism} factors this workload as
\(T_n=LR\), adds independent noise to \(Rx\), and reconstructs the prefix
answers with \(L\).  Under pure \(\eps\)-differential privacy and unit
\(\ell_1\) adjacency the relevant sensitivity of \(R\) is
\(\lVert R\rVert_{1\to1}\), the largest \(\ell_1\) norm of a column
(Lemma~\ref{lem:sensitivity}), and the Laplace mechanism calibrated to it is
\(\eps\)-differentially private (Proposition~\ref{prop:privacy}).  Two
factorization costs therefore arise:
\begin{align}
  \cfrob(T_n)
  &=
  \inf_{T_n=LR}
  \frac{\lVert L\rVert_{\mathrm F}}{\sqrt n}
  \lVert R\rVert_{1\to1},
  \label{eq:cf-def}\\
  \ctwo(T_n)
  &=
  \inf_{T_n=LR}
  \lVert L\rVert_{2\to\infty}
  \lVert R\rVert_{1\to1}.
  \label{eq:c2-def}
\end{align}
The infima range over arbitrary finite inner dimensions and arbitrary real
factors; these are the two costs of Arkhipov and
Kalinin~\cite[eqs.~(5)--(6)]{arkhipov-kalinin}.

Arkhipov and Kalinin prove upper bounds of order \((\log n)^{3/2}\) for both
costs.  They also prove matching lower bounds when both factors have entries
in \(\{0,1\}\), and their arXiv v1 closes by stating that extending the lower
bound to arbitrary matrix factorizations remains open.  The binary reduction
is based on support intersections and does not survive signed cancellation: a
real factorization can represent a zero entry by cancellation among many
nonzero products, so the arbitrary-real question requires a different
invariant.

The comparison made here is deliberately versioned and contract-specific.  We
give a self-contained proof of the arbitrary-real matrix-factorization lower
bound stated as open in Arkhipov--Kalinin v1, under
\eqref{eq:cf-def}--\eqref{eq:c2-def} exactly as printed there.  One general
step of the proof is classical and is cited rather than claimed; see
Remark~\ref{rem:classical-conversion} and Section~\ref{sec:attribution}.

\subsection{Results}

The main theorem gives the same exponent without sign, sparsity, or
squareness restrictions and with arbitrary finite inner dimension.

\begin{theorem}[Arbitrary-real prefix-factorization order]
\label{thm:main-intro}
There are absolute constants \(0<c\leq C<\infty\) such that, for every
\(n\geq1\),
\begin{equation}
  c(\log(n+1))^{3/2}
  \leq
  \cfrob(T_n)
  \leq
  \ctwo(T_n)
  \leq
  C(\log(n+1))^{3/2}.
  \label{eq:main-intro}
\end{equation}
\end{theorem}

For a differential-privacy reader the operative consequence is the following,
proved as Corollary~\ref{cor:error-order}.

\begin{corollary}[Optimized squared-error order, restated as
Corollary~\ref{cor:error-order}]
\label{cor:error-intro}
For every \(\eps>0\), within the pure-\(\eps\)-DP Laplace matrix-mechanism
class and over arbitrary finite real factors,
\begin{equation}
  \begin{aligned}
    \inf_{T_n=LR}\MaxSE(\mathcal M_{L,R},n)
    &=
    \Theta\!\left(\frac{\log^3(n+1)}{\eps^2}\right),\\
    \inf_{T_n=LR}\MeanSE(\mathcal M_{L,R},n)
    &=
    \Theta\!\left(\frac{\log^3(n+1)}{\eps^2}\right).
  \end{aligned}
  \label{eq:error-intro}
\end{equation}
\end{corollary}

The lower bound is mediated by a finite \(p\)-nuclear power.  For
\(A\in\R^{m\times n}\) and \(0<p<1\), set
\begin{equation}
  \nucpow_p(A)
  =
  \inf_{A=\sum_{q=1}^{r}u_qv_q^{\mathsf T}}
  \sum_{q=1}^{r}
  \bigl(\lVert u_q\rVert_2\lVert v_q\rVert_1\bigr)^p,
  \label{eq:nuclear-intro}
\end{equation}
where the representation is finite and exact.  Appendix~\ref{sec:nuclear}
records a quantitative finite-dimensional form of the classical
approximation-space conversion, with an explicit constant, that lower-bounds
\(\nucpow_p(A)\) by the complete profile of aggregate column widths.
Specializing it to \(T_n\) yields a phase transition at \(p=2/3\).

\begin{theorem}[Fixed-\(p\) prefix phase diagram; restated and proved as
Theorem~\ref{thm:phase}]
\label{thm:phase-intro}
For every fixed \(0<p<1\),
\begin{equation}
  \nucpow_p(T_n)
  =
  \Theta_p\!\left(
  \begin{cases}
    n, & 0<p<2/3,\\
    n\log(n+1), & p=2/3,\\
    n^{3p/2}, & 2/3<p<1.
  \end{cases}
  \right).
  \label{eq:phase-intro}
\end{equation}
\end{theorem}

The subscript in \(\Theta_p\) is not decorative: the implied constants are not
uniform in \(p\).  The sums producing the three orders in
\eqref{eq:phase-intro} carry a factor of order \(|1-3p/2|^{-1}\) and so blow up
as \(p\to2/3\) from either side, and the admissible constant
\(C_p=3/(1-2^{-(1-p)})\) of Theorem~\ref{thm:width-profile} diverges as
\(p\to1\).  Only \(p\to0\) is benign, where \(C_p\to6\).

At the critical value \(p=2/3\), Hölder's inequality converts the
\(\Omega(n\log n)\) nuclear power into
\[
  \frac{\lVert L\rVert_{\mathrm F}}{\sqrt n}
  \lVert R\rVert_{1\to1}
  =
  \Omega((\log n)^{3/2})
\]
for every real factorization \(T_n=LR\).  Since the normalized Frobenius row
energy is no larger than the maximum row energy, the same lower bound applies
to \(\ctwo\).

\subsection{The baseline that \texorpdfstring{\eqref{eq:main-intro}}{(1.4)}
improves on}
\label{sec:baseline}

It is worth recording what part of \eqref{eq:main-intro} is already classical,
so that the increment is visible.  For every matrix \(C\) one has
\(\lVert C\rVert_{1\to2}\leq\lVert C\rVert_{1\to1}\), since
\(\lVert c\rVert_2\leq\lVert c\rVert_1\) columnwise.  Hence every pair
\((L,R)\) admissible in \eqref{eq:cf-def}--\eqref{eq:c2-def} is admissible, at
a no-larger product, for the Euclidean-sensitivity
costs~\cite[eq.~(48)]{arkhipov-kalinin}
\begin{equation}
  \gtwo(M)=\inf_{M=LR}\lVert L\rVert_{2\to\infty}\lVert R\rVert_{1\to2},
  \qquad
  \gfrob(M)=\inf_{M=LR}\frac{\lVert L\rVert_{\mathrm F}}{\sqrt n}
  \lVert R\rVert_{1\to2},
  \label{eq:gamma-def}
\end{equation}
so that
\begin{equation}
  \gtwo(T_n)\leq\ctwo(T_n),
  \qquad
  \gfrob(T_n)\leq\cfrob(T_n).
  \label{eq:gamma-le-c}
\end{equation}
Both \(\gamma\)-quantities for \(T_n\) are known at order \(\log n\), with the
leading constant~\cite{fichtenberger,henzinger-uu,henzinger-ku}.  Thus
\eqref{eq:gamma-le-c} directly gives
\(\cfrob(T_n),\ctwo(T_n)=\Omega(\log(n+1))\).  A distinct lower-bound route
applies beyond matrix mechanisms: Theorem~4 of Henzinger, Upadhyay and
Upadhyay~\cite{henzinger-uu} covers \(\eps>0\) and
\(0\leq\delta<c/(2e^{\eps})\) for an absolute \(c>0\).  In particular, its
pure-DP specialization \(\delta=0\), with \(0<\eps\leq1\), gives
\(\MeanSE=\Omega(\eps^{-2}\log^2 n)\) for every continual counting mechanism;
\(\MaxSE\geq\MeanSE\) gives the same lower order for the maximum-coordinate
criterion.  These two routes are consistent, but they are not equivalent.

So the \(\Omega(\log n)\) part of \eqref{eq:main-intro} is not at issue and is
not claimed here.  What Theorem~\ref{thm:main-intro} contributes over that
baseline is exactly the additional factor \(\sqrt{\log n}\), and
\eqref{eq:gamma-le-c} localizes where it comes from: the mixed-norm gap
between \(\lVert R\rVert_{1\to1}\) and \(\lVert R\rVert_{1\to2}\).

\begin{remark}[Every normalized matrix norm is capped at
\texorpdfstring{\(\log n\)}{log n}]
\label{rem:convex-barrier}
The \(\Omega(\log n)\) ceiling in \eqref{eq:gamma-le-c} is not an artifact of
the two particular quantities \(\gtwo,\gfrob\).  Arkhipov and
Kalinin~\cite[Theorem~5.4]{arkhipov-kalinin} prove that every matrix norm
\(\gamma\) normalized by \(\gamma(st^{\mathsf T})=1\) for all
\(s,t\in\{-1,1\}^{n}\) satisfies
\(\gamma(T_n)\leq\lceil\log_2n\rceil+1\).  Any such norm used as a lower-bound
proxy for \eqref{eq:cf-def}--\eqref{eq:c2-def} is therefore capped at order
\(\log n\) and cannot witness the exponent \(3/2\).  Consistently, the same
authors show that \(\cfrob\) and \(\ctwo\) are themselves not norms: both fail
the triangle inequality~\cite[\S5]{arkhipov-kalinin}.  The lower bound proved
below accordingly runs through a quasi-normed scale --- the \(p\)-nuclear
powers \eqref{eq:nuclear-intro} at \(p=2/3<1\) --- on which no such
normalization is available, and the harmonic divergence responsible for the
extra \(\sqrt{\log n}\) occurs only because \(p<1\).
\end{remark}

\subsection{Proof architecture}

The lower bound has three steps and the upper bound one.  First, a
\(k\)-dimensional subspace cannot be close to more than half of the nested
suffix indicators: otherwise well-separated close endpoints would produce
\(2k\) disjoint long interval vectors close to the same \(k\)-plane,
contradicting Bessel's inequality.  This gives
\(D_k(T_n)\gtrsim n^{3/2}k^{-1/2}\) in the low-rank range \(1\leq k\leq n/16\)
(Appendix~\ref{sec:width}).  Second, the
classical approximation-space conversion turns \(p\)-summable rank-one
coefficients into a weighted nuclear-norm approximation profile, and for maps
\(\ell_\infty^n\to\ell_2^m\) those approximation numbers are exactly the
widths \(D_k(A)\); a direct proof is retained for its explicit constant,
removing the \(k\) largest atoms to leave mass at least \(D_k(A)\) and
accumulating the tails by a reverse Hardy inequality
(Appendix~\ref{sec:nuclear}).  Third, at \(p=2/3\) every rank scale contributes
order \(n\), producing the harmonic factor, and Hölder's inequality transfers
the critical \(p\)-nuclear bound to both costs (Appendix~\ref{sec:lower}).  A
Fenwick interval factorization then supplies matching upper bounds for both
costs and for the whole fixed-\(p\) phase diagram (Appendix~\ref{sec:upper}).

\subsection{Why these criteria and this class}
\label{sec:why}

Arkhipov and Kalinin record that under pure DP the strongest lower bounds known
for \emph{both} mean and maximum squared error are of order
\(\Omega(\eps^{-2}\log^2n)\), against an \(O(\eps^{-2}\log^3n)\) upper bound, so
that the asymptotic gap was open in both; and that ``the lack of asymptotically
better mechanisms has motivated the view that the current upper bound may be
optimal''~\cite[\S1]{arkhipov-kalinin}, a view they attribute
to~\cite{andersson-elders}.  Corollary~\ref{cor:error-intro} closes that gap
inside the pure-\(\eps\)-DP matrix-mechanism class, and closes it in the
direction that view anticipated.

The two criteria are the ones this class itself produces.
Proposition~\ref{prop:error} shows that a fixed factorization determines
exactly \eqref{eq:maxse-def} and \eqref{eq:meanse-def}, and the mean criterion
is the one Bairaktari and Larsen identify as the relevant measure for learning
applications, reserving \(\ell_\infty\) for monitoring
tasks~\cite[\S1]{bairaktari-larsen}.  They are also the two smaller quantities:
writing \(e=\mathcal M(x)-T_nx\),
\[
  \MeanSE
  \leq
  \MaxSE
  \leq
  \sup_x\mathbb E\bigl[\max_i e_i^2\bigr],
\]
so at a fixed order a lower bound on either of the first two is the stronger
statement.  The class, in turn, is where the quantitative work on this problem
is done: the leading-constant program
of~\cite{fichtenberger,henzinger-ku,andersson-elders,arkhipov-kalinin} lies
entirely inside it, and \eqref{eq:cf-def}--\eqref{eq:c2-def} are defined only
there.

The restriction to that class is real, and this framing does not conceal it.
Under pure DP the expected-\(\ell_\infty\) question remains open by a
\(\sqrt{\log n}\) factor over all mechanisms, between
\(\Omega(\eps^{-1}\log^{3/2}n)\) and
\(O(\eps^{-1}\log^{2}n)\)~\cite[Table~1]{bairaktari-larsen}, whereas
\eqref{eq:maxse-def}--\eqref{eq:meanse-def} are determined here on both sides.
Part of why the second question closes and the first does not is that the
second is asked of a smaller class of mechanisms.

\subsection{Claim boundary}

Theorem~\ref{thm:main-intro} is a result about the two costs in
\eqref{eq:cf-def}--\eqref{eq:c2-def}.  It is not a lower bound for every
continual-release mechanism.  It does not replace
\(\lVert R\rVert_{1\to1}\) by an approximate-DP sensitivity norm, and it says
nothing about learning accuracy, utility, or any dataset.  The qualitative
content of the generic width inequality is classical and is attributed in
Remark~\ref{rem:classical-conversion}; only its explicit constant and the
direct finite-dimensional proof are refinements here.  The matrix \(T_n\)
itself is classical, and is named as such in Section~\ref{sec:related}.  No
broader historical priority claim is made for the prefix-width estimate, the
fixed-\(p\) phase diagram, or the two factorization-cost bounds; the residual
limits are recorded in Section~\ref{sec:limitations}.

\section{Related work and positioning}
\label{sec:related}

\paragraph{Continual counting under pure differential privacy.}
The continual observation model, and with it the binary-tree mechanism for
counting, is due to Dwork, Naor, Pitassi and Rothblum~\cite{dwork-continual}
and, independently, Chan, Shi and Song~\cite{chan-shi-song}.  Both obtain
\(O(\eps^{-2}\log^3 n)\) squared error.  In the versioned line of work compared
here, that order remains the upper benchmark and subsequent improvements
optimize constants.  Before the general
factorization of Arkhipov and Kalinin, the strongest leading constant among
tree mechanisms was achieved by the \(k\)-ary construction with the
subtraction trick of Andersson, Pagh, Steiner and
Torkamani~\cite{andersson-elders}.  Arkhipov and
Kalinin~\cite{arkhipov-kalinin} improve those leading constants with a general
matrix factorization: they define the costs \(\cfrob\) and \(\ctwo\) in
\eqref{eq:cf-def}--\eqref{eq:c2-def}, derive the Laplace-mechanism error
formulas reproduced in Proposition~\ref{prop:error}, and lift a numerically
optimized low-dimensional factorization to all \(n\) by an explicit recursion.
The upper proof in the present paper uses a classical Fenwick factorization
rather than their optimized stacking construction; its role here is to make
the order comparison and the \(p\)-nuclear upper bound self-contained, not to
compete on constants.

\paragraph{Lower bounds.}
Building on the general factorization lower-bound framework of Edmonds,
Nikolov and Ullman~\cite{edmonds-nikolov-ullman}, Henzinger, Upadhyay and
Upadhyay~\cite[Theorem~4]{henzinger-uu} prove a
mean-squared-error lower bound for every \((\eps,\delta)\)-differentially
private continual counting mechanism when \(\eps>0\) and
\(0\leq\delta<c/(2e^{\eps})\) for an absolute \(c>0\).  At
\(\delta=0\) and \(0<\eps\leq1\), it is
\(\Omega(\eps^{-2}\log^2 n)\); the inequality
\(\MaxSE\geq\MeanSE\) transfers that order to the maximum-coordinate
criterion.  The resulting logarithmic gap to the
\(O(\eps^{-2}\log^3 n)\) pure-DP upper bound remains open for general
mechanisms under the versioned comparison in
\cite{arkhipov-kalinin}.  Within the matrix-mechanism class, Arkhipov and
Kalinin~\cite{arkhipov-kalinin} prove an
\(\Omega(\eps^{-2}\log^3 n)\) lower bound for factorizations whose factors have
entries in \(\{0,1\}\) --- a class containing the binary tree mechanism and the
\(k\)-ary tree mechanisms without the subtraction trick --- via a
skew-Bollobás set-system argument, and state the extension to arbitrary
factorizations as open.  Two further observations of theirs bound that binary
class from inside it.  Their \(0/1\) lower bound for \(\MaxSE\) carries
leading constant \(0.168\), larger than the constant \(0.0778\) achieved by
their own general factorization, so improving even the leading constant
requires leaving the binary class; and their Theorem~5.4 caps every normalized
matrix norm on \(T_n\) at order \(\log n\)
(Remark~\ref{rem:convex-barrier}), which closes the convex route to the
exponent rather than merely leaving it unexplored.  Under their contract,
Theorem~\ref{thm:main-intro} proves the corresponding order for arbitrary real
factors; it says nothing about the second half of their open problem, which
asks for a lower bound beyond the matrix mechanism altogether.

Concurrently with~\cite{arkhipov-kalinin}, Bairaktari and
Larsen~\cite[Theorem~1]{bairaktari-larsen} prove that every
\((\eps,\delta)\)-differentially private continual counting mechanism has
\[
  \max_{x\in\{0,1\}^n}\mathbb E\bigl[\lVert\mathcal M(x)-T_nx\rVert_\infty\bigr]
  =
  \Omega\!\left(\frac{\log^{3/2}n}{\max\{\eps,\delta\}}\right)
\]
for \(n^{-1+\Omega(1)}<\eps<1\) and \(0<\delta<C\) with \(C>0\) an absolute
constant.  Since a pure \(\eps\)-DP mechanism is \((\eps,\delta)\)-DP for every
\(\delta>0\), taking \(\delta\) below \(\min\{C,\eps\}\) gives
\(\Omega(\eps^{-1}\log^{3/2}n)\) for pure DP as well, as those authors
record; the bound holds for every mechanism and not only for matrix
mechanisms.  It does not subsume
Theorem~\ref{thm:main-intro}, and the reason is the error functional rather
than the mechanism class: their quantity takes the expectation \emph{after}
maximizing across coordinates, whereas \eqref{eq:maxse-def} and
\eqref{eq:meanse-def} average squared coordinate errors.  These functionals
are not ordered in general, so a lower bound on the former does not transfer
to either squared-error criterion.  The gap is genuine and not merely an
artifact of proof technique.  For \(n\) independent standard coordinates,
\(\mathbb E[\max_i|e_i|]^2\) is of order \(\log n\) while
\(\max_i\mathbb E[e_i^2]\) is of order \(1\).  In the other direction, already
in one dimension an error equal to \(L\) with probability \(L^{-2}\) and zero
otherwise has squared expected absolute error \(L^{-2}\) but mean squared
error \(1\).  The same non-transfer is recorded
by~\cite[\S1.2]{arkhipov-kalinin}, and
Remark~\ref{rem:expected-max} states the corresponding boundary here.  The two
results are complementary: theirs reaches beyond matrix mechanisms under an
expected-maximum criterion, while Theorem~\ref{thm:main-intro} determines the
arbitrary-real factorization costs that govern the coordinatewise and mean
squared errors.

\paragraph{Approximate-DP factorization norms.}
Replacing \(\ell_1\) by \(\ell_2\) sensitivity turns \(\cfrob,\ctwo\) into the
matrix norms \(\gfrob,\gtwo\) of \eqref{eq:gamma-def}.  For the counting matrix
these are settled.  Mathias's classical circulant estimate supplies the
longstanding upper bound for \(\gtwo(T_n)\)~\cite{mathias-circulant};
Fichtenberger, Henzinger and Upadhyay~\cite{fichtenberger} compute the leading
term with its constant through completely bounded norms, Henzinger, Upadhyay
and Upadhyay~\cite{henzinger-uu} give a matching singular-value lower bound
for \(\gfrob\), and Henzinger, Kalinin and Upadhyay~\cite{henzinger-ku}
sharpen the next term.  The general theory of these norms is developed by
Matoušek, Nikolov and Talwar~\cite{matousek-nikolov-talwar}.  Denisov et
al.~\cite{denisov-adaptive} also optimize matrix factorizations for adaptive
streams and private stochastic optimization; their setting uses
approximate-DP Euclidean sensitivity and is distinct from the pure-Laplace
objectives here.
Those results do not settle \eqref{eq:cf-def}--\eqref{eq:c2-def}: by
\eqref{eq:gamma-le-c} they give only the \(\Omega(\log n)\) baseline, and the
mixed-norm gap they discard is precisely what contributes the additional
square-root logarithm.

\paragraph{The prefix matrix in Banach space geometry.}
The matrix \eqref{eq:prefix-matrix} is not specific to differential privacy.
It is the coefficient matrix of the \emph{finite summation operator}
\(\Sigma_n\in\mathfrak L(\ell_1^n,\ell_\infty^n)\) of Pietsch and
Wenzel~\cite[\S0.7.3]{pietsch-wenzel}, one of the standard test operators of
that field: their \S7.6 uses it to characterise super weak compactness and
superreflexivity.  The two are the same array of scalars but not the same
normed-space operator: the \(p\)-nuclear functional
\eqref{eq:nuclear-intro} used below treats \(T_n\) as an operator
\(\ell_\infty^n\to\ell_2^n\) instead, and that realization is what fixes the
atom weight \(\lVert u_q\rVert_2\lVert v_q\rVert_1\).  Their summary table records twelve quantities for
\(\Sigma_n\): ten resolved asymptotic entries and two open rows.  Every entry
is an ideal norm relative to an orthonormal system; none is a nuclear
quantity, an approximation number, a width, or a factorization cost, and the
question asked there --- whether the \(\Sigma_n\) factor \emph{through} one
fixed operator with uniformly bounded factor norms, pursued further by
Wenzel~\cite{wenzel} --- is a different one from the growth rate measured
here.  This is positioning, not priority: the object is classical, while none
of the quantities catalogued in the monograph is the nuclear-side functional
studied below.

\paragraph{Approximation spaces and \(p\)-nuclearity.}
Pietsch~\cite{pietsch-approximation} introduced approximation schemes, their
approximation numbers, the spaces obtained by weighting those numbers, and the
Transformation Theorem that maps one such space into another when sparse
objects are sent to rank-controlled objects.  Hinrichs and
Pietsch~\cite{hinrichs-pietsch} develop the theory of \(p\)-nuclear operators
in Grothendieck's sense and, in their \S7, apply that apparatus with the
nuclear operators as ambient space and the finite-rank operators as
approximating sets; their Theorem 7.1 is the inclusion this paper needs, and
Remark~\ref{rem:classical-conversion} records the specialization and the
parameter transfer.  Neither source states the prefix-width estimate, the
fixed-\(p\) evaluation for \(T_n\), or the two factorization-cost bounds:
their worked examples are identity maps and Fourier matrices.  Of the
surrounding literature, Kwapień and Pełczyński~\cite{kwapien-pelczynski}
prove logarithmic bounds for the main triangle projection on unconditional
matrix norms --- Banach-norm statements, whereas convexifying
\eqref{eq:cf-def}--\eqref{eq:c2-def} would discard the nonconvex scale
responsible for the exponent \(3/2\), a loss that
Remark~\ref{rem:convex-barrier} records as a barrier proved in the source
compared here rather than as a stylistic preference;
Laprestè~\cite{lapreste} and
Reinov~\cite{reinov} supply classical language for mixed \(\ell_2/\ell_1\)
factorizations and finite tensor quasi-norms; and Fewster, Ojima and
Porrmann~\cite{fewster-ojima-porrmann} use the countable Grothendieck
convention that Proposition~\ref{prop:finite-countable} reconciles with the
finite one.

\section{Preliminaries}
\label{sec:prelim}

Throughout, \(\log\) denotes the natural logarithm.  The asymptotic statements
are unaffected by the base, but the explicit constants are not: the harmonic
comparison \(\sum_{k\leq m}k^{-1}\geq\log(m+1)\) used in the proof of
Theorem~\ref{thm:critical} holds for the natural logarithm and fails in
base \(2\).  The one place where a binary logarithm appears, in the statement
of \cite[Theorem~5.4]{arkhipov-kalinin} quoted in
Remark~\ref{rem:convex-barrier}, is written \(\log_2\) explicitly.

\subsection{Norms and factorization costs}

For a real matrix \(M\), write
\begin{equation}
  \lVert M\rVert_{2\to\infty}
  =
  \max_i\lVert M_{i,\cdot}\rVert_2,
  \qquad
  \lVert M\rVert_{1\to1}
  =
  \max_j\lVert M_{\cdot,j}\rVert_1.
  \label{eq:operator-norms}
\end{equation}
The former is the largest Euclidean row norm, and the latter the largest
\(\ell_1\) column norm.  The inequality
\begin{equation}
  \frac{\lVert L\rVert_{\mathrm F}}{\sqrt n}
  \leq
  \lVert L\rVert_{2\to\infty}
  \label{eq:frob-row}
\end{equation}
implies \(\cfrob(T_n)\leq\ctwo(T_n)\).

\subsection{Differential privacy and the factorization mechanism}
\label{sec:privacy}

\begin{definition}[Adjacency and pure differential privacy]
\label{def:dp}
Streams \(x,x'\in\R^{n}\) are \emph{adjacent}, written \(x\sim x'\), when
\(\lVert x-x'\rVert_1\leq1\).  For \(\eps>0\), a randomized mechanism
\(\mathcal M\) with values in \(\R^{n}\) is
\emph{\(\eps\)-differentially private} if
\begin{equation}
  \Pr[\mathcal M(x)\in S]
  \leq
  e^{\eps}\Pr[\mathcal M(x')\in S]
  \label{eq:dp}
\end{equation}
for every pair \(x\sim x'\) and every measurable \(S\subseteq\R^{n}\).
\end{definition}

\begin{lemma}[\(\ell_1\) sensitivity of a factor]
\label{lem:sensitivity}
For every \(R\in\R^{r\times n}\),
\begin{equation}
  \sup_{\lVert d\rVert_1\leq1}\lVert Rd\rVert_1
  =
  \lVert R\rVert_{1\to1}.
  \label{eq:sensitivity}
\end{equation}
\end{lemma}

\begin{proof}
Writing \(R_{\cdot,j}\) for the \(j\)-th column,
\[
  \lVert Rd\rVert_1
  =
  \Bigl\lVert\sum_{j=1}^{n}d_jR_{\cdot,j}\Bigr\rVert_1
  \leq
  \sum_{j=1}^{n}|d_j|\,\lVert R_{\cdot,j}\rVert_1
  \leq
  \lVert d\rVert_1\lVert R\rVert_{1\to1},
\]
which gives \(\leq\).  Equality holds at \(d=e_{j^{*}}\) for any column
\(j^{*}\) attaining the maximum in \eqref{eq:operator-norms}.
\end{proof}

For a factorization \(T_n=LR\) with \(R\in\R^{r\times n}\), define the additive
Laplace matrix mechanism
\begin{equation}
  \mathcal M_{L,R}(x)=L(Rx+Z),
  \label{eq:mechanism}
\end{equation}
where the coordinates of \(Z\in\R^{r}\) are independent centered Laplace
variables of scale \(b=\lVert R\rVert_{1\to1}/\eps\).  Since \(LR=T_n\neq0\) we
have \(R\neq0\), so \(b>0\).

\begin{proposition}[The factorization mechanism is \(\eps\)-DP]
\label{prop:privacy}
For every \(\eps>0\) and every factorization \(T_n=LR\), the mechanism
\eqref{eq:mechanism} is \(\eps\)-differentially private in the sense of
Definition~\ref{def:dp}.
\end{proposition}

\begin{proof}
Let \(x\sim x'\) and put \(d=x-x'\), so \(\lVert d\rVert_1\leq1\).  The random
vector \(Rx+Z\) has density
\(w\mapsto(2b)^{-r}\prod_{i=1}^{r}\exp(-|w_i-(Rx)_i|/b)\) on \(\R^{r}\), so for
every \(w\) the ratio of the densities at \(x\) and at \(x'\) is at most
\[
  \exp\!\left(\frac{\lVert Rx-Rx'\rVert_1}{b}\right)
  \leq
  \exp\!\left(\frac{\lVert R\rVert_{1\to1}\lVert d\rVert_1}{b}\right)
  \leq
  e^{\eps},
\]
using the triangle inequality, Lemma~\ref{lem:sensitivity} and the choice of
\(b\).  Integrating over any measurable set shows that \(x\mapsto Rx+Z\)
satisfies \eqref{eq:dp}; this is the Laplace
mechanism~\cite{dwork-calibrating} applied at sensitivity
\(\lVert R\rVert_{1\to1}\).  The map \(w\mapsto Lw\) is fixed and does not
depend on the stream, so composing with it preserves the guarantee.
\end{proof}

\subsection{Fixed-factorization and optimized errors}

For any randomized release \(\mathcal M\), write
\begin{align}
  \MaxSE(\mathcal M,n)
  &=
  \sup_{x}\max_{1\leq i\leq n}
  \mathbb E\!\left[
    \bigl(\mathcal M(x)_i-(T_nx)_i\bigr)^2
  \right],
  \label{eq:maxse-def}\\
  \MeanSE(\mathcal M,n)
  &=
  \sup_x
  \frac1n
  \mathbb E\!\left[
    \lVert\mathcal M(x)-T_nx\rVert_2^2
  \right].
  \label{eq:meanse-def}
\end{align}
The suprema are over the stream domain of the mechanism, which by
Definition~\ref{def:dp} is \(\R^{n}\) here; a supremum rather than a maximum
is used because on an unbounded domain a general mechanism need not attain
one.  For the additive mechanism in \eqref{eq:mechanism} the error
distribution does not depend on \(x\), so both suprema are attained and equal
the constant value.

\begin{remark}[Stream domain relative to the compared contract]
\label{rem:stream-domain}
Arkhipov and Kalinin state continual counting on binary streams
\(x\in\{0,1\}^{n}\), with the same unit \(\ell_1\) adjacency
\cite[\S1.1]{arkhipov-kalinin}.  Definition~\ref{def:dp} uses the ambient
formulation \(x\in\R^{n}\) instead, which enlarges the adjacency relation and
therefore imposes the stronger privacy requirement.  Neither the algebraic
costs nor the errors change.  By Lemma~\ref{lem:sensitivity},
\(\lVert R\rVert_{1\to1}\) is the largest \(\ell_1\) column norm and is already
attained at \(d=\pm e_j\) --- exactly the differences of adjacent binary
streams --- so the calibrated noise scale in \eqref{eq:mechanism} is the same
under either domain, and by the previous paragraph so are
\eqref{eq:maxse-def}--\eqref{eq:meanse-def}.  The factorization problem is
therefore identical, and the lower bounds proved below are statements about
\eqref{eq:cf-def}--\eqref{eq:c2-def} alone, in which no stream domain appears.
\end{remark}

\begin{proposition}[Exact squared-error formulas]
\label{prop:error}
For \(\eps>0\) and a fixed factorization \(T_n=LR\),
\begin{align}
  \MaxSE(\mathcal M_{L,R},n)
  &=
  \frac{2}{\eps^2}
  \lVert L\rVert_{2\to\infty}^{2}
  \lVert R\rVert_{1\to1}^{2},
  \label{eq:maxse-fixed}\\
  \MeanSE(\mathcal M_{L,R},n)
  &=
  \frac{2}{n\eps^2}
  \lVert L\rVert_{\mathrm F}^{2}
  \lVert R\rVert_{1\to1}^{2}.
  \label{eq:meanse-fixed}
\end{align}
Consequently, the infima of the two errors over all finite real
factorizations are
\begin{equation}
  \frac{2\ctwo(T_n)^2}{\eps^2}
  \quad\text{and}\quad
  \frac{2\cfrob(T_n)^2}{\eps^2},
  \label{eq:optimized-errors}
\end{equation}
respectively.
\end{proposition}

\begin{proof}
The additive error is \(LZ\).  A centered Laplace variable of scale \(b\)
has variance \(2b^2\).  Independence and zero means therefore give, for row
\(i\) of \(L\),
\[
  \mathbb E[(LZ)_i^2]
  =
  \frac{2\lVert R\rVert_{1\to1}^{2}}{\eps^2}
  \lVert L_{i,\cdot}\rVert_2^2.
\]
Maximizing over \(i\) gives \eqref{eq:maxse-fixed}; averaging over the \(n\)
rows gives \eqref{eq:meanse-fixed}.  Taking the corresponding infima and
using \eqref{eq:cf-def}--\eqref{eq:c2-def} proves
\eqref{eq:optimized-errors}.
\end{proof}

\begin{remark}[The two optimized errors need not be equal]
\label{rem:two-step-errors}
Equation~\eqref{eq:optimized-errors} identifies two separate infima; their
common asymptotic order does not make their finite values equal.  For the
two-step prefix matrix, Arkhipov and Kalinin prove
\(\ctwo(T_2)=\sqrt2\) and
\(\cfrob(T_2)=\sqrt{3/2}\)~\cite[Lemmas~5.2--5.3]{arkhipov-kalinin}.  Hence
\[
  \inf_{T_2=LR}\MaxSE(\mathcal M_{L,R},2)=\frac4{\eps^2},
  \qquad
  \inf_{T_2=LR}\MeanSE(\mathcal M_{L,R},2)=\frac3{\eps^2}.
\]
Corollaries~\ref{cor:error-intro} and~\ref{cor:error-order} assert only that
the two sequences have the same order as \(n\) grows.
\end{remark}

\begin{remark}[What these errors do not measure]
\label{rem:expected-max}
The quantity in \eqref{eq:maxse-fixed} is
\(\max_i\mathbb E[(LZ)_i^2]\).  It is not
\(\mathbb E[\max_i(LZ)_i^2]\), nor does it determine
\(\mathbb E[\max_i|(LZ)_i|]\).  The coordinates of \(LZ\) are generally
correlated, and bounds for either expected maximum require a separate
probabilistic argument.  The last quantity is the criterion of
\cite{bairaktari-larsen}; the two directions are compared in
Section~\ref{sec:related}.
\end{remark}

\subsection{\texorpdfstring{\(p\)}{p}-nuclear power}

\begin{definition}[Finite \(p\)-nuclear power]
\label{def:nuclear}
For \(A\in\R^{m\times n}\) and \(0<p<1\), define
\begin{equation}
  \nucpow_p(A)
  =
  \inf_{A=\sum_{q=1}^{r}u_qv_q^{\mathsf T}}
  \sum_{q=1}^{r}
  \lambda_q^p,
  \qquad
  \lambda_q=\lVert u_q\rVert_2\lVert v_q\rVert_1,
  \label{eq:nuclear-def}
\end{equation}
where \(r<\infty\).
\end{definition}

The vectors \(u_q\in\R^m\) and \(v_q\in\R^n\) define rank-one operators from
\(\ell_\infty^n\) to \(\ell_2^m\), and \(\lambda_q\) is the corresponding
operator-norm product.

The countable convention is Grothendieck's, in the form stated by Hinrichs and
Pietsch~\cite[eq.~(1.3)]{hinrichs-pietsch}: for \(0<p\leq1\), an operator
\(T:X\to Y\) between Banach spaces is \emph{\(p\)-nuclear} when it admits a
representation
\begin{equation}
  Tx=\sum_{k\geq1}\tau_k\langle x,x_k^{*}\rangle y_k,
  \qquad
  (x_k^{*})\subset B_{X^{*}},
  \quad
  (y_k)\subset B_{Y},
  \quad
  (\tau_k)\in\ell_p,
  \label{eq:grothendieck}
\end{equation}
and one sets \(\nu_p(T)=\inf\lVert(\tau_k)\rVert_{\ell_p}\) over all such
representations.  Taking \(X=\ell_\infty^n\), so that \(B_{X^{*}}\) is the
\(\ell_1\) ball, and \(Y=\ell_2^m\), a rank-one atom \(u_qv_q^{\mathsf T}\)
normalizes in both directions to
\(\tau_q=\lVert u_q\rVert_2\lVert v_q\rVert_1=\lambda_q\).  So
\(\nu_p(A)^p\) is exactly the infimum in \eqref{eq:nuclear-def} taken over
\emph{countable} representations.  The next proposition says that the
restriction to finite ones costs nothing.

\begin{proposition}[Finite and countable conventions agree]
\label{prop:finite-countable}
For every finite real matrix \(A\) and every \(0<p<1\),
\begin{equation}
  \nucpow_p(A)=\nu_p(A)^p,
  \label{eq:finite-countable}
\end{equation}
where \(\nu_p\) is the Grothendieck quasi-norm of \eqref{eq:grothendieck} for
operators \(\ell_\infty^n\to\ell_2^m\).
\end{proposition}

\begin{proof}
Every finite representation is countable, giving one inequality.  Conversely,
consider a countable representation
\begin{equation}
  A=\sum_{q\geq1}u_qv_q^{\mathsf T},
  \qquad
  \sum_{q\geq1}\lambda_q^p<\infty.
  \label{eq:countable-representation}
\end{equation}
Because \(0<p<1\), the latter summability implies
\(\sum_q\lambda_q<\infty\): after finitely many terms, \(\lambda_q\leq1\)
and hence \(\lambda_q\leq\lambda_q^p\).  Thus the partial sums converge to
\(A\) in the operator norm \(\ell_\infty^n\to\ell_2^m\).

Let \(B_N\) be the residual after \(N\) terms.  Its finite column
decomposition is
\begin{equation}
  B_N=\sum_{j=1}^{n}(B_Ne_j)e_j^{\mathsf T}.
  \label{eq:residual-columns}
\end{equation}
The additional \(p\)-cost is bounded by
\begin{equation}
  \sum_{j=1}^{n}\lVert B_Ne_j\rVert_2^p
  \leq
  n\lVert B_N\rVert_{\ell_\infty^n\to\ell_2^m}^{p}
  \longrightarrow0.
  \label{eq:residual-cost}
\end{equation}
Appending \eqref{eq:residual-columns} to the initial \(N\) atoms yields a
finite exact representation with asymptotically no extra cost.  Taking
infima proves the reverse inequality.
\end{proof}

\subsection{Aggregate column widths}

\begin{definition}[Aggregate column width]
\label{def:width}
For \(A\in\R^{m\times n}\), with columns \(a_j=Ae_j\), define
\begin{equation}
  D_k(A)
  =
  \inf_{\substack{E\subseteq\R^m\\\dim E\leq k}}
  \sum_{j=1}^{n}\dist(a_j,E)_2.
  \label{eq:width-def}
\end{equation}
Throughout, \(E\) ranges over linear subspaces of \(\R^m\) and
\begin{equation}
  \dist(a,E)_2=\inf_{e\in E}\lVert a-e\rVert_2
  \label{eq:dist-def}
\end{equation}
is the Euclidean distance from \(a\) to \(E\); the subscript records which norm
measures the distance, not a coordinate.
\end{definition}

Unlike a spectral approximation number, \(D_k(A)\) sums Euclidean distances
column by column.  This mixed aggregate geometry matches the \(\ell_1\)
coefficient norm in \eqref{eq:nuclear-def}.

\newcommand{\CandidateBProofAppendices}{%
\section{Aggregate width of the prefix chain}
\label{sec:width}

Write \(t_j=T_ne_j\).  Thus \(t_j\) is the indicator of the suffix
\(\{j,\ldots,n\}\).

\begin{lemma}[Dense endpoints force aggregate width]
\label{lem:width-lower}
If \(n\geq16\) and \(1\leq k\leq n/16\), then
\begin{equation}
  D_k(T_n)
  \geq
  \frac{1}{16\sqrt2}\frac{n^{3/2}}{\sqrt{k}}.
  \label{eq:width-lower}
\end{equation}
\end{lemma}

\begin{proof}
Fix a subspace \(E\) with \(\dim E\leq k\) and a threshold \(\theta>0\).
Suppose that at least \(n/2\) indices satisfy
\begin{equation}
  \dist(t_j,E)_2<\theta.
  \label{eq:close-column}
\end{equation}
List any \(G\geq n/2\) such indices as
\[
  g_1<\cdots<g_G.
\]
Set \(M=2k\) and
\begin{equation}
  s=\left\lfloor\frac{G}{4k}\right\rfloor.
  \label{eq:block-spacing}
\end{equation}
Since \(k\leq n/16\) and \(G\geq n/2\), one has \(G/(4k)\geq2\), so
\begin{equation}
  s\geq\frac{G}{8k}\geq\frac{n}{16k}\geq1.
  \label{eq:spacing-lower}
\end{equation}

For \(0\leq\ell<M\), choose
\begin{equation}
  \sigma_\ell=g_{2\ell s+1},
  \qquad
  \tau_\ell=g_{(2\ell+1)s+1}.
  \label{eq:endpoints}
\end{equation}
All requested indices are available: since \(2Ms=4ks\leq G\) and
\(s\geq1\),
\[
  (2M-1)s+1\leq G-s+1\leq G.
\]
The intervals
\[
  I_\ell=[\sigma_\ell,\tau_\ell)
\]
are pairwise disjoint and have cardinality at least \(s\).  Moreover,
\begin{equation}
  x_\ell=t_{\sigma_\ell}-t_{\tau_\ell}
  =\mathbf 1_{I_\ell}.
  \label{eq:interval-vector}
\end{equation}
Hence the normalized vectors
\[
  y_\ell=\frac{x_\ell}{\lVert x_\ell\rVert_2},
  \qquad 0\leq\ell<M,
\]
are orthonormal.  If \(P_E\) denotes orthogonal projection onto \(E\),
Bessel's inequality gives
\begin{equation}
  \sum_{\ell=0}^{M-1}\dist(y_\ell,E)_2^2
  =
  M-\sum_{\ell=0}^{M-1}\lVert P_Ey_\ell\rVert_2^2
  \geq M-k=k.
  \label{eq:bessel}
\end{equation}

Distance to a linear subspace is a seminorm.  Both endpoints in
\eqref{eq:endpoints} satisfy \eqref{eq:close-column}, so
\begin{equation}
  \dist(x_\ell,E)_2
  \leq
  \dist(t_{\sigma_\ell},E)_2
  +
  \dist(t_{\tau_\ell},E)_2
  <2\theta.
  \label{eq:interval-close}
\end{equation}
Since \(\lVert x_\ell\rVert_2\geq\sqrt{s}\), equations
\eqref{eq:bessel}--\eqref{eq:interval-close} imply
\begin{equation}
  k
  <
  M\frac{4\theta^2}{s}
  =
  \frac{8k\theta^2}{s}.
  \label{eq:contradiction-inequality}
\end{equation}
Choose
\[
  \theta=\sqrt{\frac{n}{128k}}.
\]
By \eqref{eq:spacing-lower}, \(\theta^2\leq s/8\), contradicting
\eqref{eq:contradiction-inequality}.  Therefore at least \(n/2\) columns
have distance at least \(\sqrt{n/(128k)}\).  Summing their distances gives
\[
  \sum_{j=1}^{n}\dist(t_j,E)_2
  \geq
  \frac n2\sqrt{\frac{n}{128k}}
  =
  \frac{1}{16\sqrt2}\frac{n^{3/2}}{\sqrt{k}}.
\]
Taking the infimum over \(E\) proves the lemma.
\end{proof}

\begin{lemma}[Matching aggregate-width upper order]
\label{lem:width-upper}
For \(1\leq k\leq n\),
\begin{equation}
  D_k(T_n)
  \leq
  \sqrt2\,\frac{n^{3/2}}{\sqrt{k}}.
  \label{eq:width-upper}
\end{equation}
Consequently,
\[
  D_k(T_n)
  =
  \Theta\!\left(\frac{n^{3/2}}{\sqrt{k}}\right)
\]
uniformly for \(n\geq16\) and \(1\leq k\leq n/16\).
\end{lemma}

\begin{proof}
Partition \(\{1,\ldots,n\}\) into \(k\) consecutive nonempty blocks, each
of size at most \(\lceil n/k\rceil\), and let \(E\) be the span of their
indicator vectors.  Each suffix indicator \(t_j\) agrees outside the block
containing \(j\) with the indicator of a union of complete blocks.  The latter
vector lies in \(E\).  Therefore
\[
  \dist(t_j,E)_2
  \leq
  \sqrt{\lceil n/k\rceil}
  \leq
  \sqrt{\frac{2n}{k}}.
\]
Summing over the \(n\) columns proves \eqref{eq:width-upper}.  The matching
lower order follows from Lemma~\ref{lem:width-lower}.
\end{proof}

\begin{remark}[Why a quadratic shortcut is insufficient]
\label{rem:rms}
A lower bound on \(\sum_j\dist(t_j,E)_2^2\) does not imply the required
lower bound on \(\sum_j\dist(t_j,E)_2\) by reversing Cauchy--Schwarz.  The
false inequality
\[
  \sum_{j=1}^{n}d_j
  \geq
  \sqrt{n\sum_{j=1}^{n}d_j^2}
\]
already fails for \((d_1,\ldots,d_n)=(1,0,\ldots,0)\) when \(n>1\).
Lemma~\ref{lem:width-lower} instead proves that a fixed positive fraction of
the columns is far from every \(k\)-plane.
\end{remark}

\section{From widths to a \texorpdfstring{\(p\)}{p}-nuclear obstruction}
\label{sec:nuclear}

\subsection{The widths are approximation numbers in the nuclear norm}

Let \(\nucideal(\ell_\infty^n,\ell_2^m)\) denote the finite-dimensional space
of nuclear operators equipped with its nuclear norm \(\nu_1\).  For
\(C:\ell_\infty^n\to\ell_2^m\),
\begin{equation}
  \nu_1(C)
  =
  \sum_{j=1}^{n}\lVert Ce_j\rVert_2.
  \label{eq:nuclear-column-norm}
\end{equation}
Indeed, the column representation
\(C=\sum_{j}(Ce_j)e_j^{\mathsf T}\) gives the upper bound; conversely, every
nuclear representation \(C=\sum_qy_qw_q^{\mathsf T}\) satisfies
\[
  \sum_j\lVert Ce_j\rVert_2
  \leq
  \sum_q\lVert y_q\rVert_2\sum_j|w_q(j)|
  =
  \sum_q\lVert y_q\rVert_2\lVert w_q\rVert_1,
\]
and taking the infimum gives the reverse bound.  Here \(\nu_1\) is the
ordinary nuclear norm; it is distinct from the \(p\)-nuclear quasi-norm
\(\nu_p\) of \eqref{eq:grothendieck}.

Let \(a_j(A\mid\nucideal)\) be the \(j\)-th approximation number of \(A\) when
the ambient norm is \(\nu_1\) and the approximating set consists of operators
of rank at most \(j-1\).  Equation \eqref{eq:nuclear-column-norm} gives the
exact identification
\begin{equation}
  a_{k+1}(A\mid\nucideal)
  =
  \inf_{\rank(B)\leq k}\nu_1(A-B)
  =
  D_k(A).
  \label{eq:width-as-approximation-number}
\end{equation}
For one inequality direction, take \(E=\operatorname{range}(B)\) and sum the
columnwise distance bounds.  For the other, take \(B=P_EA\), where \(P_E\) is
the orthogonal projection onto a \(k\)-dimensional approximating subspace, and
then take the infimum over \(E\).

\begin{remark}[The conversion below is classical]
\label{rem:classical-conversion}
Theorem~\ref{thm:width-profile} is a finite-dimensional, quantitative form of
an inclusion between operator ideals that is not due to this paper, and the
attribution is recorded here rather than in a numbered result of our own.

In the terminology of Hinrichs and Pietsch~\cite[\S7]{hinrichs-pietsch}, take
the approximation scheme whose ambient space is the nuclear operators
\(\nucideal(X,Y)\) under the nuclear norm and whose approximating sets are the
operators of rank at most \(j\), and write
\(\nucideal^{\mathrm{app}}_{s,w}:=(\nucideal)^{1/s}_{w}\) for the resulting
approximation space, with quasi-norm
\(\lVert x\mid X^{\varrho}_{w}\rVert
=\lVert(j^{\varrho-1/w}a_j(x\mid X))\mid\ell_w\rVert\).
Their Theorem 7.1 states
\begin{equation}
  \nucideal_{r,w}\subseteq\nucideal^{\mathrm{app}}_{s,w}
  \qquad\text{if }0<r<1,\quad\tfrac1s=\tfrac1r-1,\quad 0<w\leq\infty.
  \label{eq:hp-theorem}
\end{equation}
Take \(r=w=p\), so that \(\nucideal_{p,p}=\nucideal_p\) in that paper's
abbreviation and \(\varrho=1/s=1/p-1\).  The target quasi-norm of \(A\) is
then
\[
  \bigl\lVert A\mid\nucideal^{\mathrm{app}}_{s,p}\bigr\rVert
  =
  \Bigl(\sum_{j\geq1}
  \bigl[j^{(1/p-1)-1/p}a_j(A\mid\nucideal)\bigr]^{p}\Bigr)^{1/p}
  =
  \Bigl(\sum_{j\geq1}j^{-p}a_j(A\mid\nucideal)^{p}\Bigr)^{1/p},
\]
with the ambient norm and approximating class exactly those of
\eqref{eq:width-as-approximation-number}, so
\(a_j(A\mid\nucideal)=D_{j-1}(A)\).  Since \eqref{eq:hp-theorem} is a bounded
inclusion of quasi-Banach operator ideals, re-indexing by \(k=j-1\), using
\(k^{-p}\leq2^p(k+1)^{-p}\), and applying
Proposition~\ref{prop:finite-countable} yields
\eqref{eq:width-profile} below with an unspecified constant depending only on
\(p\).

The same conclusion follows from Pietsch's earlier
framework~\cite[pp.~117--118, 120, 123--124, 126]{pietsch-approximation}: the
Transformation Theorem applies with sequence exponent \(u=p\), smoothness
\(\varrho=1/p-1\) and rank-growth exponent \(1\), the sparse-sequence identity
on p.~123 supplies \((\ell_1,f_m)^{1/p-1}_{p}=\ell_{p,p}=\ell_p\), and p.~126
places the nuclear operators in the framework.

The short direct proof in Sections~\ref{sec:tails}--\ref{sec:hardy} is
retained because it keeps the finite-dimensional chain self-contained, it
supplies the explicit admissible constant \(C_p=3/(1-2^{-(1-p)})\) appearing in
the displayed lower bound, and it makes visible that cancellation is retained
inside a signed residual until a norm is applied.  Its qualitative content is
not claimed here; the unspecified constant furnished by the classical route
would suffice for every asymptotic consequence in this paper.
\end{remark}

\subsection{Rank-one tails}
\label{sec:tails}

\begin{lemma}[Every atom tail pays the aggregate width]
\label{lem:tail}
Let
\begin{equation}
  A=\sum_{q=1}^{r}u_qv_q^{\mathsf T},
  \qquad
  \lambda_q=\lVert u_q\rVert_2\lVert v_q\rVert_1,
  \label{eq:atom-decomposition}
\end{equation}
and relabel the atoms so that
\(\lambda_1\geq\cdots\geq\lambda_r\geq0\).  Then, for
\(0\leq k<r\),
\begin{equation}
  \sum_{q>k}\lambda_q\geq D_k(A).
  \label{eq:tail-width}
\end{equation}
\end{lemma}

\begin{proof}
Let \(E_k=\operatorname{span}\{u_1,\ldots,u_k\}\), so
\(\dim E_k\leq k\).  For column \(j\), the contribution of the initial \(k\)
atoms lies in \(E_k\).  Forming the exact signed residual before applying a
norm gives
\begin{align}
  D_k(A)
  &\leq
  \sum_{j=1}^{n}
  \left\lVert\sum_{q>k}u_qv_q(j)\right\rVert_2
  \notag\\
  &\leq
  \sum_{q>k}\lVert u_q\rVert_2
  \sum_{j=1}^{n}|v_q(j)|
  =
  \sum_{q>k}\lambda_q.
  \label{eq:tail-proof}
\end{align}
No sign restriction is used.
\end{proof}

\subsection{Reverse Hardy accumulation}
\label{sec:hardy}

\begin{lemma}[Tail Hardy inequality below one]
\label{lem:hardy}
For every \(0<p<1\), every finite nonincreasing sequence
\(\lambda_1\geq\cdots\geq\lambda_r\geq0\) satisfies
\begin{equation}
  \sum_{k=1}^{r-1}
  k^{-p}
  \left(\sum_{q>k}\lambda_q\right)^p
  \leq
  C_p\sum_{q=1}^{r}\lambda_q^p,
  \qquad
  C_p=\frac{3}{1-2^{-(1-p)}}.
  \label{eq:hardy}
\end{equation}
\end{lemma}

\begin{proof}
Group \(k\) into dyadic blocks \(2^s\leq k<2^{s+1}\), and define
\begin{equation}
  B_s=\sum_{q\geq2^s}\lambda_q,
  \qquad
  b_t=\sum_{2^t\leq q<2^{t+1}}\lambda_q,
  \label{eq:dyadic-tail}
\end{equation}
with all sums truncated at \(r\).  Since \(0<p<1\),
\begin{equation}
  B_s^p
  =
  \left(\sum_{t\geq s}b_t\right)^p
  \leq
  \sum_{t\geq s}b_t^p.
  \label{eq:subadditive}
\end{equation}
Also,
\begin{equation}
  \sum_{2^s\leq k<2^{s+1}}k^{-p}
  \leq2^{s(1-p)}.
  \label{eq:block-weight}
\end{equation}
For \(k\) in the \(s\)-th block,
\(\sum_{q>k}\lambda_q\leq B_s\).  Therefore
\begin{align}
  \sum_{k=1}^{r-1}
  k^{-p}\left(\sum_{q>k}\lambda_q\right)^p
  &\leq
  \sum_s2^{s(1-p)}B_s^p
  \notag\\
  &\leq
  \sum_tb_t^p\sum_{s\leq t}2^{s(1-p)}
  \notag\\
  &\leq
  \frac{1}{1-2^{-(1-p)}}
  \sum_t2^{t(1-p)}b_t^p.
  \label{eq:hardy-mid}
\end{align}
Monotonicity gives \(b_t\leq2^t\lambda_{2^t}\), hence
\begin{equation}
  2^{t(1-p)}b_t^p
  \leq
  2^t\lambda_{2^t}^p.
  \label{eq:monotone-block}
\end{equation}
For \(t\geq1\), the preceding dyadic block contains \(2^{t-1}\) terms,
each at least \(\lambda_{2^t}\), so
\begin{equation}
  2^t\lambda_{2^t}^p
  \leq
  2\sum_{q=2^{t-1}}^{2^t-1}\lambda_q^p.
  \label{eq:preceding-block}
\end{equation}
The preceding blocks are disjoint, while the \(t=0\) term contributes
\(\lambda_1^p\).  Substituting \eqref{eq:preceding-block} into
\eqref{eq:hardy-mid} proves \eqref{eq:hardy}.
\end{proof}

\begin{theorem}[Explicit finite-dimensional width obstruction]
\label{thm:width-profile}
For every finite real matrix \(A\) and \(0<p<1\),
\begin{equation}
  \boxed{
  \nucpow_p(A)
  \geq
  \frac{1}{C_p}
  \sum_{k=1}^{\rank(A)-1}
  k^{-p}D_k(A)^p,
  }
  \label{eq:width-profile}
\end{equation}
where \(C_p=3/(1-2^{-(1-p)})\) is admissible.  The sum is empty when
\(\rank(A)\leq1\).
\end{theorem}

\begin{proof}
If \(\rank(A)\leq1\), the right side is zero and the claim is immediate.
Assume \(\rank(A)\geq2\).  Fix a finite rank-one representation, delete
zero atoms, and sort the atom costs as in Lemma~\ref{lem:tail}.  If \(r\)
atoms remain, then \(r\geq\rank(A)\).  Lemmas~\ref{lem:tail}
and~\ref{lem:hardy} give
\begin{align*}
  C_p\sum_q\lambda_q^p
  &\geq
  \sum_{k=1}^{r-1}
  k^{-p}\left(\sum_{q>k}\lambda_q\right)^p\\
  &\geq
  \sum_{k=1}^{\rank(A)-1}
  k^{-p}D_k(A)^p.
\end{align*}
The right side is independent of the representation.  Taking the infimum
proves \eqref{eq:width-profile}.
\end{proof}

\section{The critical exponent and the factorization lower bound}
\label{sec:lower}

\begin{theorem}[Critical nuclear lower bound]
\label{thm:critical}
Let
\begin{equation}
  c_{\mathrm N}
  =
  \frac{1-2^{-1/3}}{144}.
  \label{eq:c-nuclear}
\end{equation}
For every \(n\geq1\),
\begin{equation}
  \nucpow_{2/3}(T_n)
  \geq
  c_{\mathrm N}n\log(n+1).
  \label{eq:critical-lower}
\end{equation}
\end{theorem}

\begin{proof}
We begin with \(n\geq32\).  Theorem~\ref{thm:width-profile} and
Lemma~\ref{lem:width-lower} imply
\begin{align}
  C_{2/3}\nucpow_{2/3}(T_n)
  &\geq
  \sum_{k=1}^{\lfloor n/16\rfloor}
  k^{-2/3}D_k(T_n)^{2/3}
  \notag\\
  &\geq
  (16\sqrt2)^{-2/3}n
  \sum_{k=1}^{\lfloor n/16\rfloor}\frac1k
  \notag\\
  &=
  \frac n8
  \sum_{k=1}^{\lfloor n/16\rfloor}\frac1k.
  \label{eq:critical-harmonic}
\end{align}
The equality uses
\((16\sqrt2)^{2/3}=(2^{9/2})^{2/3}=8\).  The harmonic bound gives
\[
  \sum_{k=1}^{\lfloor n/16\rfloor}\frac1k
  \geq
  \log\left(\left\lfloor\frac n{16}\right\rfloor+1\right)
  >
  \log\frac n{16}.
\]
For \(n\geq32\),
\begin{equation}
  \log\frac n{16}
  \geq
  \frac16\log(n+1).
  \label{eq:log-comparison}
\end{equation}
Indeed, the ratio of the left logarithm to \(\log(n+1)\) is increasing on
this range and already exceeds \(1/6\) at \(n=32\).  Combining
\eqref{eq:critical-harmonic}--\eqref{eq:log-comparison} yields
\[
  \nucpow_{2/3}(T_n)
  \geq
  \frac{n\log(n+1)}{48C_{2/3}}
  =
  c_{\mathrm N}n\log(n+1).
\]

It remains to cover \(1\leq n<32\).  For every rank-one representation,
choose a unit entry \((i,j)\) of \(T_n\).  Then
\begin{equation}
  1
  =
  |(T_n)_{ij}|
  \leq
  \sum_q|u_q(i)v_q(j)|
  \leq
  \sum_q\lVert u_q\rVert_2\lVert v_q\rVert_1
  =
  \sum_q\lambda_q.
  \label{eq:unit-entry}
\end{equation}
Subadditivity at exponent \(2/3\) gives
\begin{equation}
  1
  \leq
  \left(\sum_q\lambda_q\right)^{2/3}
  \leq
  \sum_q\lambda_q^{2/3}.
  \label{eq:unit-subadditivity}
\end{equation}
Taking the infimum yields
\(\nucpow_{2/3}(T_n)\geq1\).  Since
\(n\log(n+1)\leq31\log32<108\) and \(c_{\mathrm N}<1/108\), the stated
bound follows for the remaining dimensions.
\end{proof}

\begin{theorem}[Arbitrary-real lower bound]
\label{thm:factor-lower}
Every finite-dimensional real factorization \(T_n=LR\) satisfies
\begin{equation}
  \frac{\lVert L\rVert_{\mathrm F}}{\sqrt n}
  \lVert R\rVert_{1\to1}
  \geq
  c_{\mathrm N}^{3/2}(\log(n+1))^{3/2}.
  \label{eq:factor-lower}
\end{equation}
Consequently,
\begin{equation}
  \cfrob(T_n),\ctwo(T_n)
  \geq
  c_{\mathrm N}^{3/2}(\log(n+1))^{3/2}.
  \label{eq:both-lower}
\end{equation}
\end{theorem}

\begin{proof}
Let \(\beta=\lVert R\rVert_{1\to1}\).  Since \(LR=T_n\neq0\), one has
\(\beta>0\).  Replacing \((L,R)\) by
\((\beta L,\beta^{-1}R)\) preserves the product and the cost, so assume
\(\lVert R\rVert_{1\to1}=1\).

Write \(l_q\) for column \(q\) of \(L\) and \(r_q^{\mathsf T}\) for row
\(q\) of \(R\).  Then
\begin{equation}
  T_n=\sum_q l_qr_q^{\mathsf T}.
  \label{eq:factor-atoms}
\end{equation}
Furthermore,
\begin{equation}
  \sum_q\lVert r_q\rVert_1
  =
  \sum_{j,q}|R_{qj}|
  \leq n,
  \label{eq:row-l1-sum}
\end{equation}
because each of the \(n\) columns of \(R\) has \(\ell_1\) norm at most one.
Hölder with exponents \(3\) and \(3/2\), applied to
\(a_q=\lVert l_q\rVert_2^{2/3}\) and \(b_q=\lVert r_q\rVert_1^{2/3}\), gives
\begin{align}
  \sum_q
  \bigl(\lVert l_q\rVert_2\lVert r_q\rVert_1\bigr)^{2/3}
  &\leq
  \left(\sum_q\lVert l_q\rVert_2^2\right)^{1/3}
  \left(\sum_q\lVert r_q\rVert_1\right)^{2/3}
  \notag\\
  &\leq
  \lVert L\rVert_{\mathrm F}^{2/3}n^{2/3}.
  \label{eq:critical-holder}
\end{align}
The left side is the cost of an admissible rank-one representation.
Theorem~\ref{thm:critical} therefore gives
\[
  c_{\mathrm N}n\log(n+1)
  \leq
  \lVert L\rVert_{\mathrm F}^{2/3}n^{2/3}.
\]
Rearranging yields
\[
  \lVert L\rVert_{\mathrm F}
  \geq
  c_{\mathrm N}^{3/2}\sqrt n(\log(n+1))^{3/2}.
\]
Undoing the normalization proves \eqref{eq:factor-lower}.  Taking the
infimum proves the \(\cfrob\) lower bound, and \eqref{eq:frob-row} proves
the \(\ctwo\) lower bound.
\end{proof}

\begin{remark}[Constants]
\label{rem:constants}
The constant in Theorem~\ref{thm:factor-lower} is explicit but not optimized.
Several estimates are deliberately crude.  All main conclusions concern
asymptotic order rather than a leading constant.
\end{remark}

\begin{remark}[Countably infinite inner dimension is also covered]
\label{rem:countable-inner}
The infima \eqref{eq:cf-def}--\eqref{eq:c2-def} range over finite inner
dimension, matching the contract of~\cite{arkhipov-kalinin}.  The proof of
Theorem~\ref{thm:factor-lower} does not use that restriction.  Let \(Q\) be
countable, \(L\in\R^{n\times Q}\) and \(R\in\R^{Q\times n}\) with
\(\sum_{q\in Q}L_{iq}R_{qj}=(T_n)_{ij}\) for all \(i,j\).  Then
\eqref{eq:row-l1-sum} still holds, because it sums the \(\ell_1\) norms of the
\(n\) \emph{columns} of \(R\) whatever the number of rows;
\(\sum_q\lVert l_q\rVert_2^2=\lVert L\rVert_{\mathrm F}^2\) is unaffected, and
the claim is vacuous unless it is finite; Hölder applies to countable sums;
and the resulting rank-one representation has finite \(2/3\)-cost, so
Proposition~\ref{prop:finite-countable} converts it into a finite one without
asymptotic loss.  Hence \eqref{eq:factor-lower} holds verbatim for countable
inner dimension, and enlarging the infima in
\eqref{eq:cf-def}--\eqref{eq:c2-def} to countable inner dimension changes
neither side of \eqref{eq:main-intro}.  This is a statement about the two
matrix costs; it does not define a Laplace mechanism with countably many noise
coordinates.
\end{remark}

\section{Dyadic upper bounds and the fixed-\texorpdfstring{\(p\)}{p} phase}
\label{sec:upper}

\begin{lemma}[Fenwick interval factorization]
\label{lem:fenwick}
There is an absolute constant \(C<\infty\) such that
\begin{equation}
  \cfrob(T_n),\ctwo(T_n)
  \leq
  C(\log(n+1))^{3/2},
  \qquad
  \nucpow_{2/3}(T_n)
  \leq
  Cn\log(n+1).
  \label{eq:fenwick-bounds}
\end{equation}
\end{lemma}

\begin{proof}
Let \(N=2^h\) be the least power of two with \(N\geq n\).  For
\(1\leq q\leq N\), define
\begin{equation}
  \ell_q=2^{\nu_2(q)},
  \qquad
  I_q=\{q-\ell_q+1,\ldots,q\},
  \label{eq:fenwick-intervals}
\end{equation}
where \(\nu_2(q)\) is the largest exponent such that \(2^{\nu_2(q)}\)
divides \(q\).  Let \(R\) be the interval--point incidence matrix
\[
  R_{q,j}=\mathbf1\{j\in I_q\}.
\]
Intervals of a fixed length are disjoint.  Each point therefore belongs to
at most one interval at each of the \(h+1\) possible lengths, while point
\(1\) belongs to the intervals indexed by \(1,2,4,\ldots,N\).  Hence
\begin{equation}
  \lVert R\rVert_{1\to1}=h+1.
  \label{eq:fenwick-r-norm}
\end{equation}

Starting from a prefix endpoint \(q_0=i\) and repeatedly setting
\[
  q_{s+1}=q_s-\ell_{q_s}
\]
partitions \(\{1,\ldots,i\}\) into at most \(h+1\) of the intervals
\(I_q\).  Let row \(i\) of \(L\) indicate those intervals.  Then
\begin{equation}
  LR=T_N,
  \qquad
  \lVert L\rVert_{2\to\infty}\leq\sqrt{h+1},
  \qquad
  \frac{\lVert L\rVert_{\mathrm F}}{\sqrt N}\leq\sqrt{h+1}.
  \label{eq:fenwick-l-norms}
\end{equation}
Thus both costs for \(T_N\) are at most \((h+1)^{3/2}\).

For the nuclear estimate, let \(u_q\) be column \(q\) of \(L\) and let
\(v_q^{\mathsf T}\) be row \(q\) of \(R\).  The interval \(I_q\) can occur
only in prefix decompositions with endpoint between \(q\) and
\(q+\ell_q-1\).  Therefore
\begin{equation}
  \lVert u_q\rVert_2\leq\sqrt{\ell_q},
  \qquad
  \lVert v_q\rVert_1=\ell_q.
  \label{eq:fenwick-atom}
\end{equation}
The exact low-bit count is
\begin{equation}
  \sum_{q=1}^{N}\ell_q
  =
  N\left(1+\frac h2\right).
  \label{eq:low-bit-count}
\end{equation}
Indeed, for \(0\leq s<h\), exactly \(N/2^{s+1}\) indices have
\(\ell_q=2^s\), while \(q=N\) contributes \(\ell_N=N\).  It follows from
\eqref{eq:fenwick-atom} that
\begin{equation}
  \nucpow_{2/3}(T_N)
  \leq
  \sum_{q=1}^{N}\ell_q
  =
  N\left(1+\frac h2\right).
  \label{eq:critical-upper}
\end{equation}

There is one harmless endpoint slack in \eqref{eq:fenwick-atom}.
The interval \(I_N=\{1,\ldots,N\}\) lies entirely inside the matrix.
Its formal endpoint-occurrence range
\[
  N,\ldots,N+\ell_N-1=N,\ldots,2N-1
\]
extends beyond the available prefixes, so \(I_N\) occurs only in the
decomposition of the \(N\)-th prefix.  Consequently
\(\lVert u_N\rVert_2=1\), rather than \(\sqrt{\ell_N}\); the displayed
estimate remains a valid upper bound.

Finally, restrict \(L\) to its initial \(n\) rows and \(R\) to its initial
\(n\) columns.  Their product is \(T_n\), and no relevant row norm, column
norm, Frobenius norm, or atom cost increases.  Since \(N<2n\), changing the
Frobenius normalization from \(\sqrt N\) to \(\sqrt n\) costs less than
\(\sqrt2\).  Since \(h+1=O(\log(n+1))\), all claims follow.
\end{proof}

\begin{theorem}[Fixed-\(p\) phase diagram; restatement of
Theorem~\ref{thm:phase-intro}]
\label{thm:phase}
For every fixed \(0<p<1\),
\begin{equation}
  \nucpow_p(T_n)
  =
  \Theta_p\!\left(
  \begin{cases}
    n, & 0<p<2/3,\\
    n\log(n+1), & p=2/3,\\
    n^{3p/2}, & 2/3<p<1.
  \end{cases}
  \right).
  \label{eq:phase}
\end{equation}
\end{theorem}

\begin{proof}
For \(n\geq32\), Theorem~\ref{thm:width-profile} and
Lemma~\ref{lem:width-lower} yield
\begin{equation}
  \nucpow_p(T_n)
  \gtrsim_p
  n^{3p/2}
  \sum_{k=1}^{\lfloor n/16\rfloor}k^{-3p/2}.
  \label{eq:phase-lower}
\end{equation}
The sum has order \(n^{1-3p/2}\), \(\log(n+1)\), or \(1\), according as
\(p<2/3\), \(p=2/3\), or \(p>2/3\).  This gives the three lower orders.

For the upper bound, use the Fenwick factorization at \(N=2^h\).  By
\eqref{eq:fenwick-atom}, atom \(q\) has product cost at most
\(\ell_q^{3/2}\).  The low-bit multiplicities give
\begin{equation}
  \nucpow_p(T_N)
  \leq
  \frac N2
  \sum_{s=0}^{h-1}2^{s(3p/2-1)}
  +
  N^{3p/2}.
  \label{eq:phase-upper}
\end{equation}
The geometric sum has the three orders in \eqref{eq:phase}.  Restricting
from \(N<2n\) to \(n\) preserves those orders.  For the finitely many
\(n<32\), the unit-entry argument in
\eqref{eq:unit-entry}--\eqref{eq:unit-subadditivity}, with exponent \(p\),
gives \(\nucpow_p(T_n)\geq1\), and the constants may be adjusted.
\end{proof}

\begin{proof}[Proof of Theorem~\ref{thm:main-intro}]
The lower bounds are Theorem~\ref{thm:factor-lower}.  The upper bounds are
Lemma~\ref{lem:fenwick}, and the middle inequality is
\eqref{eq:frob-row}.
\end{proof}

\begin{corollary}[Optimized squared-error order; restatement of
Corollary~\ref{cor:error-intro}]
\label{cor:error-order}
For every \(\eps>0\), within the pure-\(\eps\)-DP Laplace matrix-mechanism
class and over arbitrary finite real factors,
\begin{align*}
  \inf_{T_n=LR}\MaxSE(\mathcal M_{L,R},n)
  &=
  \Theta\!\left(\frac{\log^3(n+1)}{\eps^2}\right),\\
  \inf_{T_n=LR}\MeanSE(\mathcal M_{L,R},n)
  &=
  \Theta\!\left(\frac{\log^3(n+1)}{\eps^2}\right).
\end{align*}
\end{corollary}

\begin{proof}
Every mechanism \(\mathcal M_{L,R}\) is \(\eps\)-differentially private by
Proposition~\ref{prop:privacy}.  Proposition~\ref{prop:error} identifies the
two optimized errors as \(2\ctwo(T_n)^2/\eps^2\) and
\(2\cfrob(T_n)^2/\eps^2\), and Theorem~\ref{thm:main-intro} evaluates both
costs at \(\Theta((\log(n+1))^{3/2})\).
\end{proof}

}

\section{Interpretation and limitations}
\label{sec:limitations}

\subsection{What the exponent boundary says}

Theorem~\ref{thm:main-intro} rules out an exponent improvement obtained
solely by allowing signed, dense, rectangular factors of arbitrary inner
dimension.  What replaces the support-intersection step available for binary
factors is that the prefix-specific input measures all low-rank aggregate
widths simultaneously, and that cancellation stays inside an exact signed
residual until a norm is applied.

The critical value \(p=2/3\) is forced by the interaction of three scales: in
the low-rank range \(1\leq k\leq n/16\),
\[
  D_k(T_n)\asymp n^{3/2}k^{-1/2},
  \qquad
  k^{-p}D_k(T_n)^p
  \asymp
  n^{3p/2}k^{-3p/2},
\]
and the series becomes harmonic exactly when \(3p/2=1\).  That range is the
one Lemmas~\ref{lem:width-lower}--\ref{lem:width-upper} establish and the only
one the proof uses; the two-sided estimate does not extend to all \(k\leq n\),
since the columns of \(T_n\) are linearly independent and hence
\(D_n(T_n)=0\).  Hölder then
converts the \(2/3\)-power into a \(3/2\) exponent for the factorization
cost.

\subsection{Mechanism and privacy scope}

The result concerns the factorization class in
\eqref{eq:mechanism}.  It does not prove a lower bound for every interactive
or non-matrix continual mechanism, and it therefore leaves untouched the
second half of the open problem of~\cite{arkhipov-kalinin}.  It is also
specific to pure \(\eps\)-differential privacy through the \(\ell_1\)
sensitivity \(\lVert R\rVert_{1\to1}\) of Lemma~\ref{lem:sensitivity};
approximate-DP mechanisms governed by \(\ell_2\) sensitivity have a different
factorization geometry, and their order for \(T_n\) is
\(\log n\) rather than \(\log^{3/2}n\)~\cite{fichtenberger,henzinger-uu}.

Corollary~\ref{cor:error-order} concerns the maximum of coordinatewise
expected squared error and the mean expected squared error.  By
Remark~\ref{rem:expected-max}, it does not determine an expected maximum
across coordinates; that criterion is the one bounded
by~\cite{bairaktari-larsen}, whose result is neither implied by nor implies
Theorem~\ref{thm:main-intro}.  No empirical or learning-theoretic conclusion
follows from these matrix costs.

\subsection{Attribution}
\label{sec:attribution}

The conversion from \(p\)-summable rank-one coefficients to a weighted
nuclear-norm approximation profile belongs to the approximation-space theory
of the Jena school.  It is a finite-dimensional specialization of Pietsch's
Transformation Theorem~\cite{pietsch-approximation}, and Theorem 7.1 of
Hinrichs and Pietsch~\cite{hinrichs-pietsch} states the corresponding
operator-ideal inclusion directly; Remark~\ref{rem:classical-conversion}
records both reductions with their parameter transfers.
Theorem~\ref{thm:width-profile} contributes the explicit constant \(C_p\) and
a self-contained finite-dimensional proof, nothing more.  The matrix \(T_n\)
is likewise classical: it is the coefficient matrix of the finite summation
operator \(\Sigma_n:\ell_1^n\to\ell_\infty^n\) of Pietsch and
Wenzel~\cite[\S0.7.3]{pietsch-wenzel}, read here instead as an operator
\(\ell_\infty^n\to\ell_2^n\).

The prefix-dependent steps proved in this paper are the width estimate of
Lemma~\ref{lem:width-lower}, its evaluation at every fixed \(p\) in
Theorem~\ref{thm:phase}, and the transfer to \(\cfrob\) and \(\ctwo\) in
Theorem~\ref{thm:factor-lower}.  This describes the proof decomposition, not a
claim that the statements are absent from all earlier literature.  The
comparison with~\cite{arkhipov-kalinin} is versioned and contract-specific:
that source states the arbitrary-factor lower bound as open, and
Theorem~\ref{thm:main-intro} supplies it under the same two cost definitions.

\section{Conclusion}

For the prefix workload, arbitrary signed real factorizations have the same
\((\log(n+1))^{3/2}\) asymptotic cost as the known upper constructions, under
both the normalized Frobenius and the maximum-row objective.  For
\(\eps>0\), within the finite-dimensional pure-\(\eps\)-DP Laplace
matrix-mechanism class the resulting optimized maximum and mean squared errors
have order \(\Theta(\eps^{-2}\log^3(n+1))\).  Relative to the exact contract
stated as open in Arkhipov--Kalinin v1, the matrix lower bound covers arbitrary
signed, dense, rectangular factors of arbitrary finite inner dimension and, as
a cost inequality only, countable inner dimension by
Remark~\ref{rem:countable-inner}.

The proof combines a prefix-specific aggregate-width estimate with a
classical approximation-space engine: the transformation theorem of Pietsch,
or the corresponding operator-ideal inclusion of Hinrichs and Pietsch, supplies
the generic conversion from \(p\)-summable atomic coefficients to a weighted
nuclear-norm approximation profile, and the direct tail and reverse-Hardy
argument included here supplies an explicit constant.  At the critical
exponent \(2/3\), the prefix profile produces a harmonic lower bound, and
Hölder transfers it to factorization energy.  The same analysis yields
matching fixed-\(p\) asymptotics for the finite prefix operator, together with
a self-contained Fenwick upper construction.

Extending the lower bound beyond matrix mechanisms, closing the remaining
\(\log n\) gap between \(\Omega(\eps^{-2}\log^2 n)\) for general pure-DP
continual counting and \(\Theta(\eps^{-2}\log^3 n)\) for the matrix class, and
sharpening the constants are separate questions.

\clearpage
\appendix
\CandidateBProofAppendices

\end{document}